\documentclass[preprint,12pt]{elsarticle}

\usepackage{changepage}
\usepackage{dcolumn}
\usepackage[table]{xcolor}
\usepackage{floatrow}
\usepackage{floatrow}   
\usepackage{subcaption}
\usepackage{graphicx,xcolor} 
\usepackage{multirow}
\usepackage{url}

\usepackage[margin=1in]{geometry}
\usepackage{amsmath}

\usepackage{amsthm}

\usepackage{amssymb}
\usepackage{verbatim}
\usepackage{physics}
\usepackage{MnSymbol}
\usepackage{wasysym}
\usepackage{cleveref}
\usepackage{ragged2e}

\newtheorem{theorem}{Theorem}

\usepackage{lineno}

\journal{Physica A: Statistical Mechanics and its Applications}

\begin{document}

\begin{frontmatter}


\title{Evolution of cooperation in networked heterogeneous fluctuating environments}


\author{Viktor Stojkoski$^{1,2}$}
\author{Marko Karbevski$^{2,3,4,5}$}
\author{Zoran Utkovski$^{6}$}
\author{Lasko Basnarkov$^{7}$}
\author{Ljupco Kocarev$^{2,7}$}

\date{\today}

\address{
$^{1}$ SS. Cyril and Methodius University, Faculty of Economics,, blvd. Goce Delcev 9V, 1000 Skopje, North Macedonia}
\address{
$^{2}$ Macedonian Academy of Sciences and Arts, P.O. Box 428, 1000 Skopje, North Macedonia}
\address{
$^{3}$Sorsix International, Dame Gruev 18, Skopje 1000, Republic of Macedonia}
\address{
$^{4}$Institute of Mathematics,
Faculty of Natural Sciences and Mathematics,
Ss. Cyril and Methodius University,
Arhimedova 3, 1000 Skopje, Republic of Macedonia}
\address{
$^{5}$Sorbonne Université, 4 Place Jussieu, 75005 Paris, France}
\address{
$^{6}$Fraunhofer Heinrich Hertz Institute, Einsteinufer 37, 10587, Berlin, Germany}% 
\address{
$^{7}$SS. Cyril and Methodius University, Faculty of Computer Science and Engineering,  P.O. Box 393, 1000 Skopje, Macedonia}%

\begin{abstract}

Fluctuating environments are situations where the spatio-temporal stochasticity plays a significant role in the evolutionary dynamics. The study of the evolution of cooperation in these environments typically assumes a homogeneous, well mixed population, whose constituents are endowed with identical capabilities. In this paper, we generalize these results by developing a systematic study for the cooperation dynamics in fluctuating environments under the consideration of structured, heterogeneous populations with individual entities subjected to general behavioral rules. Considering complex network topologies, and a behavioral rule based on generalized reciprocity, we perform a detailed analysis of the effect of the underlying interaction structure on the evolutionary stability of cooperation. We find that, in the presence of environmental fluctuations, the cooperation dynamics can lead to the creation of multiple network components, each with distinct evolutionary properties. This is paralleled to the freezing state in the Random Energy Model. We utilize this result to examine the applicability of our generalized reciprocity behavioral  rule in a variety of settings. We thereby show that the introduced rule leads to steady state cooperative behavior that is always greater than or equal to the one predicted by the evolutionary stability analysis of unconditional cooperation. As a consequence, the implementation of our results may go beyond explaining the evolution of cooperation. In particular, they can be directly applied in domains that deal with the development of artificial systems able to adequately mimic reality, such as reinforcement learning. 
\end{abstract}


\end{frontmatter}

%\linenumbers

\section{Introduction}

Cooperation is ubiquitous in nature and is essential
to the functioning of a large number of biological systems. This phenomenon has played a fundamental role in many of
the major transitions in biological evolution, and cooperative interactions are required for many levels of
networked organizations ranging from single cells to groups of
animals and, ultimately, humans. In evolutionary biology, theoretical models provide evidence that cooperative
behavior can evolve and persist if a certain mechanism is at work, for example, kin selection, group (multilevel) selection, and different forms of reciprocity~\cite{Nowak-2006five,utkovski2017promoting}. In addition, it has been recognized that the network structure plays an important
role in the emergence of cooperation, as 
underlying network properties such as network heterogeneity, scale-freeness, etc., crucially determine the outcome of multiple
dynamical phenomena.

A standard approach for examining the effect of different mechanisms on the cooperation dynamics in complex
networks is through evolutionary game theory where the individual entities interacting in a network
are given a set of strategies that they can choose from,
and a set of payoffs (changes in the individual resource
endowment) that result from interactions with other entities based on their chosen strategies. Traditional game theory has predominantly focused on additive payoffs, which yields the world of two-player iterative social dilemmas, such as the iterated prisoner's dilemma, or multi-player (i.e. public goods) games~\cite{stojkoski2018cooperation}.

In contrast to interaction models with additive payoffs, the temporal evolution of the individual resource endowments in natural and social systems is often characterized by a multiplicative process~\cite{meder2019ergodicity}.  In evolutionary biology and ecology, these systems are studied within the general framework of \textit{fluctuating environments} , which refers to spatio-temporal stochasticity in the environmental conditions~\cite{saether2015concept}. Environmental fluctuations can be integrated in models for evolutionary games in various ways, including stochastic network structures~\cite{antonioni2011network,melamed2016strong,zimmermann2004coevolution,jordan2013contagion}, fluctuating population size~\cite{behar2014coexistence,houchmandzadeh2015fluctuation,houchmandzadeh2012selection,huang2015stochastic,gokhale2016eco,constable2016demographic} and random payoffs~\cite{stollmeier2018unfair,peters2015evolutionary,yaari2010cooperation,liebmann2017sharing}. In our concept, we follow the last strand of works, and model environmental fluctuations as randomly varying payoff values.

In environments subjected to fluctuations in the payoff structure, the non-ergodicity of the fluctuation-generating process may have non-trivial effects on the evolutionary dynamics~\cite{saether2015concept}. In particular, when the growth in the resource endowments is governed by a multiplicative process, the resulting fluctuations exhibit a net-negative effect on the time-averaged growth of the resources of the individual entities, while having no effect on the ensemble growth rate. Since, on the long run the ensemble properties are never observed, the time-averaged growth is the only relevant quantity for evolutionary performance~\cite{peters2016evaluating,peters2019ergodicity}. Recently, it has been discovered that this yields an evolutionary behavior which essentially differs from the one observed in standard models, i.e., leads to evolutionary dynamics  where cooperators can coexist with defectors without any additional auxiliary mechanism at place~\cite{stollmeier2018unfair}. For example, on this basis, it has been hypothesized that repeated pooling and sharing of resources which previously exhibit a fluctuating growth may constitute a fundamental mechanism for the evolution of cooperation in a well-mixed population. The rationale is that, by reducing the amplitude of fluctuations, pooling and sharing increases the steady state growth rate at which the cooperating entities self-reproduce~\cite{yaari2010cooperation,liebmann2017sharing,peters2015evolutionary}.

The research on the evolution of cooperation in fluctuating environments typically focuses on the following simplifying assumptions: i) the population is homogeneous, i.e. all entities are characterized by the same  individual traits, ii) the interaction takes place in a well-mixed population, and iii) the individual entities can decide only between two behavioral strategies. The first assumption constraints each entity to have the same physical characteristics and to be subjected to the same fluctuation process. In reality, however, entities are \textit{heterogeneous}: they are endowed with different capabilities and experience non-identical randomness. These traits can represent the level of skills or income profiles in economic societies~\cite{guvenen2007learning,gabaix2016dynamics} or sex differences in development, physiology, morphology and/or life-history in biological systems~\cite{forsman2018role}.
The second assumption implies that entities spend, on average, an equal time interacting with each other member of the population. In practice, individual entities are most often interacting through a social \textit{network} of contacts (or within a neighborhood in a biological network), thus favoring interactions with a certain group. It has been recognized that such network structures crucially determine the outcome of a multitude of dynamical phenomena, including cooperation~\cite{Dorogovtsev-2008}. The last assumption reduces the actions of the individual entities to a set of two possible outcomes. In other words, this assumption  does not allow for learning to take place in the population. While all these assumptions allowed for powerful analytical findings, they can significantly reduce practical applications.

In this paper, we revisit the above assumptions and develop a systematic approach for investigating the cooperation dynamics in fluctuating environments by considering structured, heterogeneous populations with individual entities subjected to more general behavioral rules. We consider an interaction model on a complex network based on pooling and sharing of resources undergoing multiplicative growth that models environmental fluctuations. We know from~\cite{yaari2010cooperation,liebmann2017sharing}, that in a well-mixed homogeneous population, cooperation is always favored under this model. When moving from a well-mixed population to a more general network structure, in Ref.~\cite{stojkoski2019cooperation} it was shown that cooperation remains the sole evolutionary stable strategy (given that the population is homogeneous),  but the performance of the population is uniquely determined by the underlying interaction topology. Complementing these works, here we investigate the joint effects of heterogeneous population and network topology. A crucial observation from the analysis is that, when population heterogeneity is introduced to the model, complex behavior emerges and the stability of cooperation is dependent on both the network topology and the behavioral update rules employed by the individual entities.

By analyzing the model from an evolutionary perspective, we derive the criteria required for unconditional cooperation to be evolutionary stable within any interaction structure. More importantly, by applying these criteria to simple and tractable analytical examples, as well as to more complex numerical examples, we are able to show that the interaction structure may yield a non-trivial effect on the behavior in multiplicative dynamics. In particular, the interaction structure can induce creation of multiple connected components, each consisting of unconditional cooperators. Even though the multiplicative nature of the process makes the distribution of resources among the cooperating entities to have a an infinite average in the time limit, the connected components instigate dynamics in which some unconditional cooperators are much ``richer'' than others. This effect is related to the freezing transition in the Random Energy model of Derrida~\cite{derrida1980random}, and in our model occurs as a result of the heterogeneous individual capabilities, in combination with the network structure~\cite{gueudre2014explore}.

The evolutionary stability analysis assumes that each individual entity implements strategies that allow only for unconditional cooperator or unconditional defector behavior. As such it does not allow us to examine the effect of the behavioural learning update rules employed by the individual entities. Due to the non-ergodicity, which arises as a consequence of the fluctuating environment, standard rules cannot be easily translated to our model~\cite{peters2016evaluating}. Instead, one should use a suitable modification of the rule which is studied. 

Here we develop a simple behavioral update and examine its implications. The introduced rule is related to the concept of \textit{generalized reciprocity}, which itself is rooted in the principle of ``help anyone if helped by someone''~\cite{taborsky2016correlated}. In particular, under our rule the individual entities search for their optimal strategies solely by comparing their observed growth rate with the one that is expected under unconditional defection. From a game-theoretic perspective the rule may be framed in the context of continuous games, where entities may be able to chose from a continuum of strategies~\cite{utkovski2017promoting}, or as a stochastic decision making process, with random payoffs determined by the entity's experience~\cite{stojkoski2018cooperation}. Direct parallels can be made to state-of-the-art reinforcement learning techniques based on novelty search~\cite{lehman2008exploiting}. We thereby show that our rule leads to steady state cooperative behavior that is always greater than or equal to the one predicted by evolutionary stability analysis of unconditional cooperation. In fact, for a certain regime of parameter values we are able to analytically solve the model and show that then the strict inequality holds.

The rest of the paper is organized as follows. In Section~\ref{sec:model} we introduce the system model, together with a brief summary of the properties of pooling and sharing in a heterogeneous environment. In Section~\ref{sec:evolutionary-stability} we study the model from the perspective of evolutionary game theory, and derive the criteria for evolutionary stability of unconditional cooperation on complex networks. Here, we also solve several analytically tractable examples that elucidate the role of the network structure. We also examine numerically the extent to which different complex network topologies/structures promote cooperative behavior. In Section~\ref{sec:generalized-reciprocity} we present a behavioral strategy based on generalized reciprocity and study the properties of the resulting cooperation dynamics. The last section briefly summarizes our findings.

\section{Model}
\label{sec:model}

\subsection{Interaction structure}

In our model we consider a population of $N$ individual entities interacting through $M$ pools. The pools may represent various compositions that  facilitate the sharing of resources. For instance, in biological systems the pools may constitute structures that ease the transport of nutrients~\cite{roper2013cooperatively}. Similarly, in economics the pools may be seen as equivalent to baskets that are  used to share commodities in community-supported agriculture~\cite{sapir2004agenda}. In the model, the interaction structure is described by a connected bipartite random graph, whose adjacency matrix $\mathbf{B}$, with binary edge variables $B_{im} \in \left\{0,1 \right\}$, captures the participation of individual $i$ in pool $m$ ($\mathrm{B}_{im} = 1$, indicating participation of $i$ in pool $m$).
The bipartite graph representation, although not essential for our concept, offers a principled way of capturing wider information regarding the group composition and network interactions~\cite{perc2013evolutionary,gomez2011disentangling,pena2012bipartite} than standard unipartite graphs, where each node behaves as both an individual entity and a pool through which resources are shared~\cite{perc2013evolutionary}. As a result, bipartite graphs are a convenient tool for modeling pooling and sharing on networks.  
In particular, each unipartite graph can be mapped into a bipartite graph by considering a replacement graph procedure, while the opposite is not true in general, as illustrated in Fig.~\ref{fig:bipartite-network}. In this representation each node behaves as both an individual entity and a pool through which resources are shared~\cite{perc2013evolutionary}. 
We refer to~\cite{perc2013evolutionary,pena2012bipartite,stojkoski2019cooperation,tomassini2020public} for more details on the procedure for constructing a replacement graph.

We provide an example for this property in Fig.~\ref{fig:bipartite-network} where we draw two bipartite graphs, each composed of three entities. In addition, the first graph is composed of two pools (Fig.~\ref{fig:bipartite-network}a), whereas in the second graph there is an additional third pool (Fig.~\ref{fig:bipartite-network}b). In the second row of the figure we give the simple graph representation. For the first bipartite graph there is no such representation, but the second graph can be represented in a standard representation that allows for self-loops. Both bipartite graphs can be described as a unipartite only if we include more complex replacement graph representation that allows for edge directions and weights (bottom row of the figure).

\begin{figure}[t!]
\includegraphics[width=8.7cm]{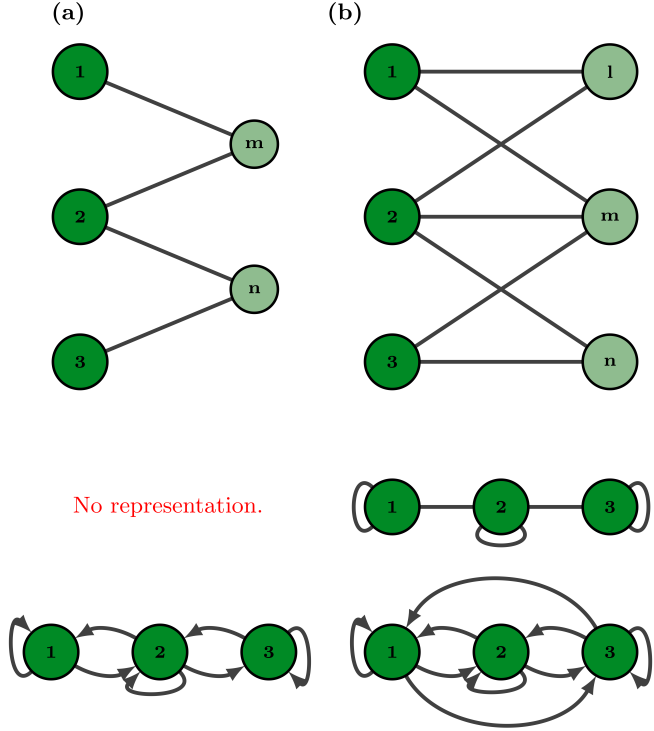}
\caption{\textbf{Bipartite graph representation.} \textbf{(a)} First row: a bipartite graph consisting of tree entities and two pools; Second row: its unipartite representation; Third row: the weighted replacement graph. \textbf{(b)} Same as \textbf{(a)}, only now there are three pools and the connections between the entities and pools are slightly different.  \label{fig:bipartite-network}}
\end{figure}

We consider discrete-time dynamics and track the evolution of the individual resource endowments $y_i(t)$, for all $i\in\mathcal{N}$. The resource growth is governed by a multiplicative process, where in each period the resource growth of the $i$-th entity is modeled by a drift $\mu_i$ and a noise amplitude $\sigma_i$. Note that we consider heterogeneous population where $\mu_i$ and $\sigma_i$ differ across the different members of the population. After the growth phase, the individual resources are a subject to a pooling phase where each entity $i$ distributes a fraction $p_i(t) / d_i$ of its resources to each of the $d_i = \sum_m \mathrm{B}_{im}$ pools where it participates. Finally, every pool shares a fraction $1/u_m$, with $u_m = \sum_i \mathrm{B}_{im}$, of the total pooled resources to each of its members. Formally, the evolution of the individual resources can be expressed as 
\begin{align}
y_i(t+\Delta t) &=\underbrace{\sum_j A_{ij} p_j(t) y_j(t) \left[ \Delta t + \mu_j \Delta t + \sigma_j \varepsilon_j(t) \sqrt{\Delta t} \right]  - p_i(t) y_i(t) \Delta t}_{\text{resources due to pooling and sharing}}\nonumber \\ &+ \underbrace{y_i(t) + \left(1 - p_i(t) \right) y_i(t) \left[ \mu_i \Delta t + \sigma_i \varepsilon_i(t) \sqrt{\Delta t} \right]}_{\text{individual resources}},
\label{eq:network-discrete}
\end{align}
where $\varepsilon_j(t)$ is a standardized Gaussian random variable and $\mathbf{A}$ is a transition matrix of the network with entries $\mathrm{A}_{ij} = \sum_m^M  \frac{B_{im}}{u_m} \frac{B_{jm}}{d_j}$ which determine the total allocated resources from entity $j$ to entity $i$. The variable $p_i(t) \in \left[ 0,1 \right]$ models the level of cooperation displayed by an individual entity in a particular period. The case of $p_i(t) = 0$ for all $t$ represents the situation when the entity always defects, i.e., the unconditional defector case. On the other hand, the case of $p_i(t) = 1$ for all $t$ is the unconditional cooperator setting. In a more general setting, the entities propensity too cooperate may vary over time based on their environment and thus $p_i(t)$ may be interpreted as an internal state that captures the propensity of an entity to engage in pooling and sharing of its resources.

Following~\cite{peters2015evolutionary},  we write Eq.~(\ref{eq:network-discrete}) in a continuous form by setting $\Delta t \to 0$
\begin{align}
\mathrm{d} y_i &= \left[\sum_j \mathrm{A}_{ij} p_j y_j - p_i y_i\right] \mathrm{d}t + \sum_j \mathrm{A}_{ij} p_j y_j \left( \mu_j \mathrm{d}t+ \sigma_j \mathrm{d}W_j  \right) + \left( 1 - p_i \right) y_i \left(\mu_i\mathrm{d}t + \sigma_i \mathrm{d}W_i \right),
\label{eq:network-differential}
\end{align}
where $(\mathrm{d}W_i)_{ i \in \mathcal{N}}$ are independent Wiener increments, i.e. $W_i(t) =\int_0^t \mathrm{d}W_i$. This representation allows us to utilize techniques such as It\^{o} calculus for the theoretical analysis of the model. As in~\cite{peters2015evolutionary}, when working with numerical estimations, we will resort the discrete formulation, with the note that in that case the  temporal differences will be rescaled such that $\Delta t = 1$. This, as we will see, allows for an easy Markov representation of the discrete model.

The resulting interaction structure is directly related to games of public goods on networks. The main difference is that in this model the fluctuations arise inherently from the individual traits owned by the entities, whereas in a public goods game the randomness is a feature of each pool. In other words, in public goods games the stochasticity comes as a result of the process of pooling, here there is randomness even without pooling~\cite{perc2013evolutionary,santos2008social,stojkoski2018cooperation}. Another similar model is the Bouchaud–Mezard model of wealth reallocation and its extensions~\cite{bouchaud2000wealth,bouchaud2015growth,gueudre2014explore}. The essential difference is that in our description, reallocation is done after the growth phase.

\subsection{Growth rate}
In a situation with additive resource dynamics, the relevant quantity for measuring the payoff of an individual entity is the steady state time-averaged absolute change of wealth~\cite{peters2019ergodicity}. This is not the case with our model, since the fluctuations introduced by the multiplicative process governing the resource dynamics make the system non-stationary, and hence non-ergodic. Under these circumstances, the relevant quantity for measuring the payoff is the steady state time-averaged growth rate of wealth, $g_i(y_i(t), t)$. In evolutionary biology, this observable is more known as the \textit{geometric mean fitness} for the accumulated payoff (i.e. resources) of a particular phenotype~\cite{saether2015concept}. Formally, it is defined as
\begin{align}
  g_i(y_i(t), t) &= \frac{1}{t}\log\left(\frac{y_i(t)}{y_i(0)}\right),
\end{align}
where $y_i(0)$ is the initial amount of resources. In what follows, we are going to assume $y_i(0)=1$. 

When the internal state for cooperation of each entity $i$ is $p_i^* = 0$, we have the unconditional defector situation and the solution to the model represents $N$ independent geometric Brownian motion (GBM) trajectories. Then, the steady state time-average growth rate can be easily found using It\^{o} calculus as
\begin{align}
g^{ \mathcal{D} }_i & = \lim_{t \to \infty} g^{\mathcal{D} }_i(y_i(t), t) = \mu_i - \frac{\sigma_i^2}{2},
\label{eq:defector-average-growth}
\end{align}
where the superscript $\mathcal{D}$ denotes the case in which no individual entity pools its resources. 
%that no individual entity pools its resources. 

On the other hand, when every entity cooperates unconditionally, i.e. $p_i(t) =1 $ for all $i$ and $t$, the growth rate of each entity converges to the same value,
\begin{align}
   g^{ \mathcal{C} } = \lim_{t \to \infty} g_i(y_i(t), t) = \langle \mu v \rangle_{\mathcal{N}} - \frac{1}{2N} \langle v^2 \sigma^2 \rangle_{\mathcal{N}},
\label{eq:full-network-coop}
\end{align}
where the superscript $\mathcal{C}$ denotes that every entity pools its resources, $\langle \cdot \rangle_{\mathcal{N}}$ is the population average and $v$ is an index given by the right-eigenvector associated with the largest eigenvalue of $\mathbf{A}$ normalized in a way such that $\sum_i v_i = N$. 

For cooperation to be evolutionary favored by the population, the quantity in Eq.~(\ref{eq:full-network-coop}) should be larger than the one of Eq.~(\ref{eq:defector-average-growth}). This is always the case if the population is homogeneous, i.e. the individual entities are subject the same drift and noise amplitude~\cite{stojkoski2019cooperation}. In a heterogeneous population where the individual entities exhibit different traits (i.e. they are subject to different drifts and noise amplitudes), the evolution of cooperation is dependent on both the network structure and the individual traits. We elaborate on this in more detail in the sequel.

\section{Evolutionary stability}
\label{sec:evolutionary-stability}

\subsection{Preliminaries}
We begin the analysis by considering the circumstance in which every individual entity can behave either as an unconditional cooperator ($p_i(t) = p_i^* = 1$ for all $t$), or as an unconditional defector ($p_i(t) = p_i^* = 0$ for all $t$), and examine the Evolutionary Stable State (ESS). As discussed in Nowak~\cite{Nowak-2006five}, an ESS of unconditional cooperation is a situation in which a large population of cooperators cannot be invaded by defectors under deterministic selection dynamics. ESS has been widely used for determining the level of cooperative behavior in various interaction structures~\cite{Perc-2017}.

In a simple interaction structure consisting of two entities $i$ and $j$, the possible outcomes are described by a payoff matrix
  \begin{align}
    \setlength{\extrarowheight}{2pt}
    \begin{tabular}{cc|c|c|}
      & \multicolumn{1}{c}{} & \multicolumn{2}{c}{ $j$}\\
      & \multicolumn{1}{c}{} & \multicolumn{1}{c}{$C$}  & \multicolumn{1}{c}{$D$} \\\cline{3-4}
      \multirow{2}*{ $i$}  & $C$ & $\pi_{CC}$ & $\pi_{CD}$ \\\cline{3-4}
      & $D$  & $\pi_{DC}$ & $\pi_{DD}$ \\\cline{3-4}
    \end{tabular}
  \end{align}
where $\pi_{ab}$ denotes the payoff of entity $i$ under strategy $a$ when entity $j$ has chosen strategy $b$. As discussed, in our case the strategies $p_i^* = 1$ and $p_i^* = 0$ translate to being an unconditional cooperator ($C$), respectively defector ($D$), and the growth rates of the entities translate to the corresponding payoffs in the payoff matrix.  
This implies that unconditional cooperation is ESS if $\pi_{CC} > \pi_{DC}$. In interaction structures involving multiple entities, unconditional cooperation is ESS if the overall benefit of cooperation for an individual entity is greater than the benefit of defection, assuming that all other entities cooperate unconditionally. In other words, in the model with networked pooling and sharing, full cooperation will be ESS if and only if
\begin{align}
    g^{\mathcal{C}} > \max_i g^{\mathcal{D}}_i.
\end{align}

\subsection{Growth rate}

In a more general setting, unconditional cooperation may only be favored by a certain subset $\mathcal{C}_l$ of the population, whereas for the other part it is optimal to behave as unconditional defectors. This intermediate case can lead to the emergence of multiple connected components where each component has its own growth rate depending on the network parameters. Recall, a connected component is a subset of the network in which any two nodes (entities) are connected to each other by paths, and which is connected to no additional nodes.

To provide a better intuition on this behavior, in Fig.~\ref{fig:network-cases} we illustrate three situations which lead to different growth rate outcomes by considering a simple network composed of five entities and two pools. In particular, in Fig.~\ref{fig:network-cases}a we assume that unconditional cooperation is evolutionary stable for all entities, Therefore, in this case the growth rate of each individual entity converges to~(\ref{eq:cooperative-growth-rate}). On the other hand, in, Fig.~\ref{fig:network-cases}b we set the growth rate of entity $3$, which acts as bridge by being the only one to pool its resources in both pools, to be large enough so as this entity behaves as an unconditional defector. This creates two separate components $\mathcal{C}_1 = \left\{1,2\right\}$ and $\mathcal{C}_2 = \left\{4,5\right\}$, which pool and share resources only between themselves. Notice that the entities also pool part of their resources to entity $3$ but they are not shared back. This implies that the resource dynamics, and hence growth rates, between the components are independent. If the growth rates of each component are different, this leads to great discrepancies in the observed resources between the cooperators belonging to separate components, with the resources of the entities in one component being negligible in comparison to the resources of the entities in the other component. This is a distinguishing feature of our model and occurs solely as a result of the heterogeneous individual capabilities, in combination with the network structure. The observed effect is formally equivalent to the freezing state in the Random Energy model of Derrida which has been extensively studied in the statistical mechanics community~\cite{derrida1980random}.
Finally, in the last example, besides entity $3$ we also set the entities in component $\mathcal{C}_2$ to be defectors. Then, there is only one connected component and each unconditional cooperator will have the same steady state growth rate. Hence, in this circumstance there is no freezing state among the entities belonging to the unconditional cooperator set. 

\begin{figure}[t!]
\includegraphics[width=10.7cm]{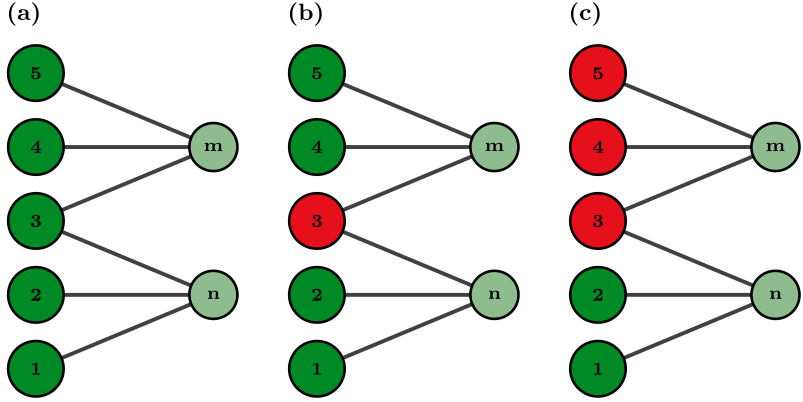}
\caption{\textbf{Creation of cooperative components.} \textbf{(a-c)} Network composed of five entities and two pools. Green nodes are the entities which prefer unconditional cooperation, whereas the red nodes depict entities which favor defection. \label{fig:network-cases}}
\end{figure}

To analyze the steady state growth rate $g_i^*$ of the entities in component $\mathcal{C}_l$ we utilize a mean-field approach together with It\^{o} calculus, as is done in~\cite{stojkoski2019cooperation}. The mean-field approach assures that each connected component of cooperators $\mathcal{C}_l$ will have a convergent growth rate given by the growth rate of the average $\langle y \rangle_{\mathcal{C}_l}$ constructed by the resource trajectories of the entities belonging to this component. 

To prove this claim, we define the rescaled resources of each cooperating entity $i$ as $\hat{y}_i(t) = \frac{y_i(t)}{\langle y \rangle_{\mathcal{C}_l}}$. Due to the normalization, in steady state, this observable also converges to a constant value $\hat{y}^*_i$ that is greater than zero and less then the number of unconditional cooperators. Hence, the growth rate of each entity $i$ in the set of unconditional cooperators can be written as
\begin{align*}
\lim_{t \to \infty} g_i( y_i(t), t) &= \lim_{t \to \infty} \frac{1}{t} \log\left(\frac{y_i(t)}{y_i(0)} \right) \\
&=  \lim_{t \to \infty} \frac{1}{t} \log \left( \langle y(t) \rangle_{\mathcal{C}_l} \cdot \hat{y}_i(t) \right) \\
&= \lim_{t \to \infty} \frac{1}{t} \log \left( \langle y(t) \rangle_{\mathcal{C}_l} \right) + \lim_{t \to \infty} \frac{1}{t} \log \left( \hat{y}_i(t) \right) \\
&= \lim_{t \to \infty} \frac{1}{t} \log \left(\langle y(t) \rangle_{\mathcal{C}_l} \right) \\
&\doteq \lim_{t \to \infty} g(\langle y(t) \rangle_{\mathcal{C}_l}, t).
\end{align*}

Consequently, one can use It\^{o}'s lemma to directly calculate the cooperative time-average growth rate. Formally, the lemma states that the differential of an arbitrary one-dimensional function $f(\mathbf{y},t)$ governed by an It\^{o} drift-diffusion process is given by
\begin{align}
\mathrm{d}f(\mathbf{y},t) &= \frac{\partial f}{\partial t} \mathrm{d}t + \sum_i \frac{\partial f}{\partial y_i}\mathrm{d}y_i + \frac{1}{2}\sum_i \sum_j \frac{\partial^2 f}{\partial y_i \partial y_j} \mathrm{d}y_i \mathrm{d}y_j.
\label{eq:ito-lemma}
\end{align}
In the case of $g(\langle y \rangle_{\mathcal{C}_l},t)$, we have that $f(t,\mathbf{y}) = \log(\langle y \rangle_{\mathcal{C}_l})$. Then, the first and second derivative of $f$ with respect to $y_i$ and $y_j$ are easily calculated as $\frac{\partial f}{\partial y_i}  = \frac{1}{N_{\mathcal{C}_l}} \frac{1}{\langle y \rangle_{\mathcal{C}_l}}$ and $\frac{\partial^2 f}{\partial y_i \partial y_j} = - \frac{1}{N^2_{\mathcal{C}_l}} \frac{1}{\langle y \rangle_{\mathcal{C}_l}^2}$, where $N_{\mathcal{C}_l}$ is the number of cooperators in steady state. Moreover, this transformation makes the differential $\mathrm{d}f(\mathbf{y},t)$ ergodic, and since we are looking at steady state averages, $\mathrm{d}y_i$ and $\mathrm{d}y_i \mathrm{d}y_j$ can be substituted with their expected values $\langle \mathrm{d}y_i \rangle$ and $\langle \mathrm{d}y_i \mathrm{d}y_j \rangle$. To estimate these expectations we utilize the independent Wiener increment property $\langle \mathrm{d}W_i^2 \rangle = \mathrm{d}t$, and define $z_j^{\left[\mathcal{C}_l\right]} = \sum_{k \in \mathcal{C}_l} A_{kj} $. Further, we omit terms of order $\mathrm{d}t^2$ as they are negligible. As a result, we obtain that
\begin{align*}
\langle \mathrm{d}y_i \rangle = \left[\sum_{k \in \mathcal{C}_l} A_{ik}(1+ \mu_k) y_k - y_i\right] \mathrm{d}t,  
\end{align*}
and,
\begin{align*}
 \langle \mathrm{d}y_i \mathrm{d}y_j \rangle = \sum_{k \in \mathcal{C}_l} A_{ik} A_{jk} \sigma^2_k y_k^2 \mathrm{d}t .    
\end{align*}
By inserting the estimates in Eq.~(\ref{eq:ito-lemma}) we can approximate the time-average growth rate as
\begin{align}
g(\langle y \rangle_{\mathcal{C}_l},t) &= \frac{1}{N_{\mathcal{C}_l}} \sum_{i \in N_{\mathcal{C}_l}} \hat{y}_i(t) \left[(1 + \mu_i) z_i^{\left[\mathcal{C}_l\right]} - 1 \right] - \frac{1}{2} \frac{1}{N_{\mathcal{C}_l}^2} \sum_{i \in \mathcal{C}_l} (z_i^{\left[\mathcal{C}_l\right]} \hat{y}_i(t) )^2 \sigma_i^2.
\label{eq:cooperative-growth-rate}
\end{align}
Eq.~(\ref{eq:cooperative-growth-rate}) describes the growth of $\langle y \rangle_{\mathcal{C}_l}$ as a function of $\hat{y}_i(t)$. To derive its steady state behavior, notice that we can remove the unconditional defectors from the cooperative dynamics, set $\Delta t = 1$, and rewrite Eq.~(\ref{eq:network-discrete}) as
\begin{align}
 y_i(t+1 ) \approx \sum_{j \in \mathcal{C}_l} A_{ij} y_j(t) \left[ 1 + \mu_j + \sigma_j \varepsilon_j(t) \right].
\label{eq:cooperative-dynamics}    
\end{align}
When Eq.~(\ref{eq:cooperative-dynamics}) is divided by the population average resources $\langle y(t+1) \rangle_{\mathcal{C}_l}$ and is written in matrix form, the steady state dynamics for the rescaled resources can be approximated as
\begin{align}
\mathbf{\hat{y}^*}  &= \lim_{t \to \infty} \mathbf{\hat{y}}(t) \approx \lim_{t \to \infty} \frac{\mathbf{A}_{\mathcal{C}_l} \mathbf{y}(t)}{\langle \mathbf{A}_{\mathcal{C}_l} \mathbf{y}(t)\rangle_{\mathcal{C}_l}},
\label{eq:markov-chain}
\end{align}
where $\mathbf{A}_{\mathcal{C}_l}$ is a reduced version of the transition matrix $\mathbf{A}$ in which includes only the rows and columns associated to the entities in the set $\mathcal{C}_l$. 
By the power method, this leads to
\begin{align}
 \hat{y}_i^* &= v^{\left[\mathcal{C}_l\right]}_i,
\label{eq:vi}
\end{align}
where $v^{\left[\mathcal{C}_l\right]}_i$ is the $i$-th element of the right-eigenvector of $\mathbf{A}_{\mathcal{C}_l}$ associated with the largest eigenvalue $\lambda_{\mathcal{C}_l}$  normalized in a way such that $\sum_i v^{\left[\mathcal{C}_l\right]}_i = N_{\mathcal{C}_l}$.

By inserting the estimates of Eq.~(\ref{eq:vi}) in Eq.~(\ref{eq:cooperative-growth-rate}) we get that the steady state cooperative growth rate is
\begin{align}
g_{\mathcal{C}_l} &= \left( \lambda_{\mathcal{C}_l} -1 \right) + \langle x^{\left[ \mathcal{C}_l\right]} \mu \rangle_{\mathcal{C}_l} - \frac{1}{2} \frac{1}{N_{\mathcal{C}_l}} \langle (x^{\left[\mathcal{C}_l\right]})^2 \sigma^2 \rangle_{\mathcal{C}_l},
\label{eq:steady-state-cooperative-growth-rate}
\end{align}
where $x_i^{\left[\mathcal{C}_l\right]} = v_i^{\left[\mathcal{C}_l\right]} z_i^{\left[\mathcal{C}_l\right]}$, is a network centrality index whose relation with the drifts and amplitudes ultimately determines the cooperative growth rate. 

\subsection{Determining ESS \label{sec:iterative-method}} 

To determine the ESS in situations when there is a possibility that unconditional defectors and cooperators may coexist, we consider an alternate-projection method~\cite{boyd2003alternating}. In particular, notice that the ESS conditions may be reformulated as 
\begin{align*}
\mathrm{g}^*_i &= g_{\mathcal{C}_l}, \hspace{0.05cm} i \in \mathcal{C}_l, \nonumber \\ 
0 &= \left(g_{\mathcal{C}_l} - g_i^{\mathcal{D}}\right) \left( 1 - \mathrm{p}^*_i \right), 
\end{align*} 
resulting in a nonlinear system of $2N$ equations with $2N$ variables ($g_i^*$ and $p_i^*$) in total. Hence, a simplified, iterative approach based on the alternate projection method can be used for finding the steady state solution. We follow~\cite{utkovski2017promoting} and summarize the method as follows:
\begin{enumerate}
\item Set $\mathrm{p}^*_i = 0$ and $g^*_i = g_i^{\mathcal{D}}$ for all $i$ not satisfying the condition~(\ref{eq:full-network-coop}). Set $\mathrm{p}^*_i = 1$ for the remaining entities. To estimate their growth rate solve Eq.~(\ref{eq:steady-state-cooperative-growth-rate}) for each component of unconditional cooperators $\mathcal{C}_l$ and set for each $i \in \mathcal{C}_l$, $g_i = g_{\mathcal{C}_l}$.

\item For all $i$ satisfying $\mathrm{g}^*_i < g_i^{\mathcal{D}}$ in the obtained solution, set $\mathrm{p}^*_i = 0$ and $\mathrm{g}^*_i = g_i^{\mathcal{D}}$. For each remaining component of cooperators, solve again the corresponding growth rate.
\item Repeat steps $1.$ and $2.$ until there are no $\mathrm{g}^*_i < g_i^{\mathcal{D}}$.
\end{enumerate}

\subsection{Simple examples}
To provide an illustrative representation for the evolutionary properties of the model in what follows we study three simple examples. The purpose of the first example is to show the evolutionary dynamics in the simplest population structure consisting of only two entities. The second example extends the population size to an arbitrary size of $N$ entities. The last example shows how the position of the entities in an ordinary network structure can impact the individual payoff.

\paragraph{Example 1:} The first situation that we consider is a replication of the example studied in Ref.~\cite{peters2015evolutionary}. Concretely, we assume an interaction structure of two entities $i$ and $j$ who have the option to share their resources through one pool $m$, as illustrated in the right panel of Fig.~\ref{fig:simple-example1}. 

\begin{figure}[t!]
\includegraphics[width=12cm]{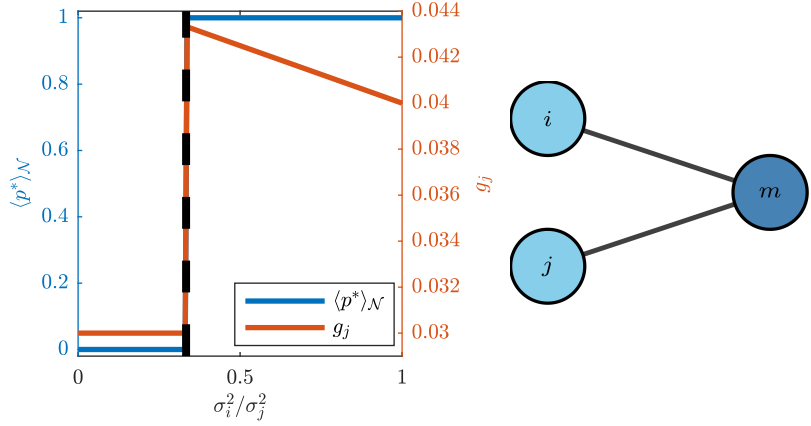}
\caption{\textbf{Example 1.} Left panel: average fraction of unconditional cooperators $\langle p^* \rangle_{\mathcal{N}}$ and steady state growth rate $g_j$ of entity $j$ as a function of the noise amplitude ratio $\sigma^2_i / \sigma^2_j$. The dashed line is the threshold after which cooperation is evolutionary stable. The right panel depicts the interaction structure. \label{fig:simple-example1}}
\end{figure}

For simplicity, we are going to assume that $i$ has a larger individual steady state growth rate, (i.e. $g^{\mathcal{D}}_i > g^{\mathcal{D}}_j$ and examine the situation from the point of view of entity $i$, In this case, the payoff matrix for $i$ reads
\begin{align}
    \setlength{\extrarowheight}{2pt}
    \begin{tabular}{cc|c|c|}
      & \multicolumn{1}{c}{} & \multicolumn{2}{c}{ $j$}\\
      & \multicolumn{1}{c}{} & \multicolumn{1}{c}{$C$}  & \multicolumn{1}{c}{$D$} \\\cline{3-4}
      \multirow{2}*{ $i$}  & $C$ & $\langle \mu \rangle_{\mathcal{N}} - \frac{\langle \sigma^2 \rangle_{\mathcal{N}}}{4}$ & 0 \\\cline{3-4}
      & $D$ & $ \mu_i - \frac{\sigma_i^2}{2}$ & $ \mu_i - \frac{\sigma_i^2}{2}$ \\\cline{3-4}
    \end{tabular}
\end{align}
which, after some reordering, implies that full unconditional cooperation is ESS if
\begin{align*}
    \mu_j - \frac{\sigma^2_j}{4} > \mu_i - \frac{3 \sigma^2_i}{4}.
\end{align*}
In the special case when $\mu_j = \mu_i$, the condition reduces to $3 \sigma_i > \sigma_j$. This result is displayed in Fig.~\ref{fig:simple-example1} where we plot the average steady state propensities for cooperation within the population and the steady state growth $g_j$ of entity $j$  as a function of the square of the amplitude ratio $\sigma_i / \sigma_j$. We observe that at the critical point at $\sigma^2_i / \sigma^2_j = 1/3$ the growth rate of entity $j$ is largest and afterwards it is decreasing linearly due to the increase in the noise amplitude $\sigma_i$ of entity $i$.

\paragraph{Example 2:} The second example extends the interaction to an arbitrary number of $N$ entities, as depicted in the right panel of Fig.~\ref{fig:simple-example2}. This is the well-mixed situation which has been extensively utilized for determining the performance of a particular mechanism in the evolution of cooperation. As such this allows us to link our model to previous findings.

\begin{figure}[t!]
\includegraphics[width=12cm]{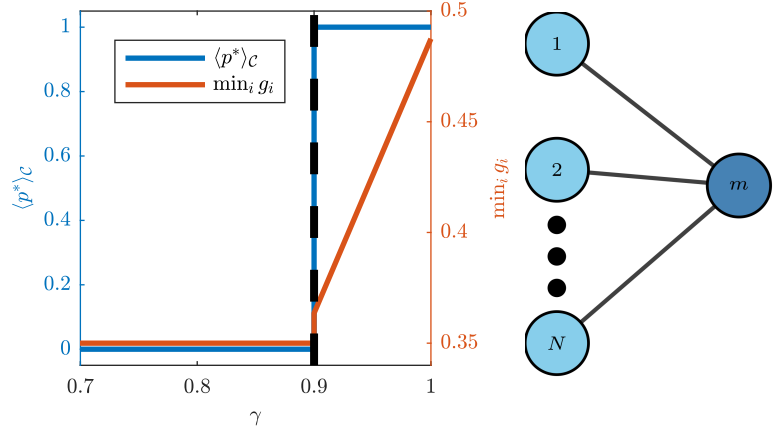}
\caption{\textbf{Example 2.} Left panel: average fraction of unconditional cooperators $\langle p^* \rangle_{\mathcal{C}}$ and minimum steady state growth rate for the entities in the unconditional cooperator set $\mathcal{C}$ as a function of the fraction of population belonging to that set $\gamma$. The dashed line is the threshold for cooperation to be ESS among the entities in $\mathcal{C}$, whereas the right panel describes the interactions. In this case $N = 12$, $\mu = 0.5$ and $\sigma^2 = 0.3$.  \label{fig:simple-example2}}
\end{figure}

To ease the analysis, we assume that a fraction of the population has an individual growth rate larger than $g^{\left[ \mathcal{N}\right]}$ and therefore the entities belonging to this group behave as unconditional defectors. The other fraction, $\gamma$, has an equal drift $\mu$ and noise amplitude $\sigma$. 

As a consequence, in this example there is only one component of possible unconditional cooperators $\mathcal{C}$ and for each entity $i$ in it we have that $v_i^{\mathcal{C}} = 1$ and $z_i^{\left[\mathcal{C}\right]} = \lambda_{\mathcal{C}} = \gamma $. Hence, the cooperative growth rate of this component is
\begin{align}
    g^{\mathcal{C}} &= \gamma - 1 + \gamma \mu - \gamma \frac{\sigma^2}{2N}.
\label{eq:well-mixed-growth}
\end{align}
Clearly, under evolutionary dynamics, unconditional cooperation by the entities in the potential cooperator set will be favored only if
\begin{align*}
    g^{\mathcal{C}} > \mu - \frac{\sigma^2}{2},
%\label{eq:well-mixed-growth-}
\end{align*}
which, when rearranged in terms of $\gamma$, yields
\begin{align}
\gamma &> \frac{1+\mu- \frac{\sigma^2}{2}}{1+ \mu- \frac{\sigma^2}{2N}}.
\label{eq:well-mixed-condition}
\end{align}

Fig.~\ref{fig:simple-example2} visualizes the dependence of the minimum steady state growth rate of the potential unconditional cooperators as a function of the fraction of potential cooperators $\gamma$. It is easily noticed that after condition~(\ref{eq:well-mixed-condition}) is satisfied, the growth rate of the entities in the set $\mathcal{C}$ increases, which means that cooperation is an ESS for them.

Eq.~(\ref{eq:well-mixed-condition}) is a direct  pooling and sharing counterpart to the evolutionary rules for cooperation that have been both numerically and analytically studied~\cite{Nowak-2006five}. In particular, notice that when there is no noise ($\sigma = 0$), then unconditional cooperation is never favored. As such, the derived inequality is directly related to the result in~\cite{stollmeier2018unfair} where the evolutionary dynamics in fluctuating environments were studied in a broader setting. The main difference is that our model allows us to easily infer the role of network structures in fluctuating environments as well as to study the presence of different behavioral update rules.

\paragraph{Example 3:} In the last example we once more examine a structure consisting of $N$ entities which now do not construct a well-mixed population and instead interact through $2$ pools, $m$ and $n$. A fixed number of entities $N-2$ interact solely through pool $m$ and one entity, $j$, interacts only through pool $n$. In addition, there is one entity, $i$, which connects the network by pooling its resources in both $m$ and $n$. This interaction structure is described in the right panel of Fig.~\ref{fig:simple-example3}.

\begin{figure}[t!]
\includegraphics[width=12cm]{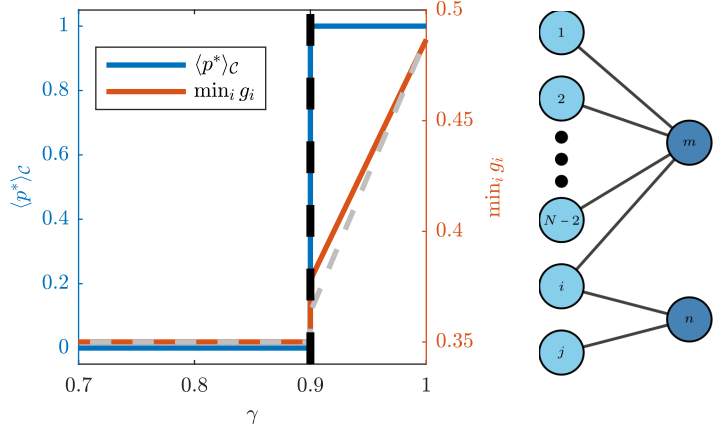}
\caption{\textbf{Example 3.} Left panel: average fraction of unconditional cooperators $\langle p^* \rangle_{\mathcal{C}}$ and minimum steady state growth rate for the entities in the cooperator set $\mathcal{C}$ as a function of the fraction of potential unconditional cooperators $\gamma$. The dashed black line is the threshold for cooperation to be ESS among the entities in $\mathcal{C}$, whereas the grey line in the background gives the growth rate for the corresponding well-mixed system. The right panel describes the interactions. In this example  $N = 12$, $\mu = 0.5$ and $\sigma^2 = 0.3$. \label{fig:simple-example3}}
\end{figure}

We assume that a fraction $1-\gamma$ of the individuals interacting solely through pool $m$ have a large enough growth rate to behave as an unconditional defectors, and thus study the circumstances under which the remaining $N_\mathcal{C} = \gamma(N-2) + 2$ entities behave as unconditional cooperators $\mathcal{C}$. In this case the interaction matrix reads
\begin{align*}
\mathbf{A}_{\mathcal{C}} &= 
\begin{bmatrix} 
\frac{1}{N-1} & \frac{1}{N-1}  & \dots & \frac{1}{N-1}  & \frac{1}{2(N-1)} & 0  \\ 
\vdots & \vdots & & \vdots & \vdots & \vdots \\
\frac{1}{N-1} & \frac{1}{N-1}  & \dots & \frac{1}{N-1}  & \frac{1}{2(N-1)} & 0 \\ 
\frac{1}{N-1} & \frac{1}{N-1}  & \dots & \frac{1}{N-1}  & \frac{1}{2(N-1)}+ \frac{1}{4}& \frac{1}{2}  \\ 
0 & 0  & \dots & 0  & \frac{1}{4}& \frac{1}{2} \\
\end{bmatrix} .
\end{align*}

Moreover, the largest eigenvalue of $\mathbf{A}_{\mathcal{C}}$ is
\begin{align}
    \lambda_{\mathcal{C}} &\approx \frac{|4\gamma -3| + 4 \gamma +3}{8},
    \label{eq:eigenvalue-3example}
\end{align}
with corresponding right-eigenvector entries $v^{\left[\mathcal{C}\right]}_k\approx \frac{(\gamma(N-2) + 2)w^{\left[\mathcal{C}\right]}_k}{\sum_l w^{\left[\mathcal{C}\right]}_l}$ approximated for large $N$ and a fixed $\gamma>0$ by 
\begin{equation}
w^{\left[\mathcal{C}\right]}_k =
\begin{cases}
\frac{|4\gamma -3|+ 4\gamma -1}{2},  & \text{if } k = i, \\\
1  & \text{if } k = j, \\\
\frac{ |4\gamma-3|(N-2)+4\gamma(N-4) + (11-3N)}{2 (N-2) }& \text{otherwise},
\end{cases}
\label{eq:v-3example}
\end{equation}
and $z^{\left[\mathcal{C}\right]}$ index entries
\begin{equation}
z^{\left[\mathcal{C}\right]}_k = 
\begin{cases}
\frac{\gamma(N-2) +1}{2(N-1)}+ \frac{1}{2} & \text{if}\ k=i, \\\
 1 & \text{if}\ k=j, \\\
\frac{\gamma(N-2)+1}{N-1},  & \text{otherwise}.
\end{cases}
\label{eq:z-3example}
\end{equation}
To illustrate the differences between this example and Example 2., we set each entity belonging to the potential set of unconditional cooperators an equal growth rate, $\mu_k = \mu$ and noise amplitude $\sigma_k = \sigma$.

Combining Eqs.~\eqref{eq:eigenvalue-3example},~\eqref{eq:v-3example} and~\eqref{eq:z-3example}, one can derive an analytical expression for the cooperative growth rate~\eqref{eq:cooperative-growth-rate}. However, the resulting mathematical expression is rather complicated. Instead of writing it down, we visualize it in Fig.~\ref{fig:simple-example3}. The figure shows the cooperative growth rate as a function of $\gamma$. In the same figure, with the light gray line we plot the cooperative growth rate of the corresponding well-mixed system (Example 2). We observe that once unconditional cooperation is achieved by the entities in $\mathcal{C}$, the presence of a network structure effectively increases the growth rate of each individual, as compared to the well-mixed situation. \textcolor{red}{Moreover, even if the population size increases indefinitely, i.e., if we take the limit as $N \to \infty$ of Eqs.~\eqref{eq:eigenvalue-3example},~\eqref{eq:v-3example} and~\eqref{eq:z-3example}, the cooperative growth rate will remain dependent on the network structure and the fraction of unconditional cooperators $\gamma$. This can be easily deduced from  Eq.~\eqref{eq:eigenvalue-3example}, as the equation is not dependent on the population size, but rather on the fraction of potential unconditional cooperators which interact solely through pool $m$ and the non-trivial interaction between entities $i$ and $j$.} This serves as an intuitive example that the interaction structure may have a disproportionate effect on the observed individual behavior.

\subsection{Random graphs}

In a more general, random graph structure, besides the size of the unconditional defector set, the positioning of the entities in the network also determines the ESS for each other entity. While Eq.~(\ref{eq:cooperative-growth-rate}) and the alternate projection method give the solution for the level of cooperation in any arbitrary situation, the extent to which different random graphs are able to support cooperation in the presence of defectors can not be deduced explicitly.

To provide an insight on this phenomena, we consider four different types of random graph models i) Random $d$-regular (RR) graph~\cite{bollobas2013modern}, ii) Erdos-Renyi (ER) random graph~\cite{erdos1960evolution}, iii) Watts-Strogatz (WS) random graph~\cite{watts1998collective}, and iv) Barabasi-Albert (BA) random scale-free network~\cite{barabasi1999emergence} and study their robustness in the presence of arbitrary positioned and/or number of defectors. In an RR graph each entity is characterized with the same degree $d$, whereas in an ER graph two entities share an edge with probability $d/N$. Both types of random graphs yield homogeneous degree distributions and low clustering coefficients. The WS random graph, on the other hand is an extension of the two aforementioned graphs which is able to capture higher levels of clustering. In fact, this random graph type lies in-between the ER and RR graph as its creation starts with an initial structure of an RR graph and then each edge is re-wired with a fixed probability. Finally, the BA graph is constructed through a preferential attachment mechanism for generating random graphs. As such it yields a power law degree distribution coupled with high clustering.

To assess the robustness of the types of random graphs we conduct the following experiment. We begin by generating a random graph through the typical algorithms that are used for this procedure. Afterwards, we choose a random fraction $1-\gamma$ of the population. We set the growth rate of the entities in it to be large enough so as they behave as unconditional defectors. We characterize the entities in the other set, the potential unconditional cooperators, with an equal drift $\mu$ and noise amplitude $\sigma$. Under these circumstances, we estimate the steady state growth rates of each entity $i$ belonging to the cooperative cluster $\mathcal{C}_l$ as
\begin{align}
g^{\mathcal{C}_l} &= \left( \lambda_{\mathcal{C}_l} -1 \right) + \mu  \langle x^{\left[ \mathcal{C}_l\right]} \rangle_{\mathcal{C}_l}  - \frac{\sigma^2}{2 N_{\mathcal{C}_l}} \langle (x^{\left[\mathcal{C}_l\right]})^2 \rangle_{\mathcal{C}_l} \nonumber \\
&= \left(\mu  + 1 \right) \lambda_{\mathcal{C}_l} - \frac{\sigma^2}{2 N_{\mathcal{C}_l}} \langle (x^{\left[\mathcal{C}_l\right]})^2 \rangle_{\mathcal{C}_l} - 1,
\label{eq:steady-state-cooperative-growth-rate-randomg-graphs}
\end{align}
using the the procedure described in Section~\ref{sec:iterative-method}. We gather the average number of unconditional cooperators among the entities in the potential unconditional cooperator set, $\langle p^* \rangle_{\mathcal{C}}$. To get the typical behavior of a particular random graph we average across random graph instances and across defector selection samples.

Numerical results are summarized in Fig.~\ref{fig:complex-networks}a where we plot the average $p^*_i$ among the entities in the cooperator set as a function of $\gamma$. We observe that for every $\gamma$ the RR graph presents itself as the most supportive for cooperation as it requires the lowest amount of entities in the unconditional cooperator set for cooperation to exist. It is followed by the WS graph, while under this criterion the ER and BA graphs have similar performance and are the least robust random graphs. 

The reason for this behavior becomes apparent if we look at the evolution of $\lambda_{\mathcal{C}}$ and $\langle (x^{\left[\mathcal{C}_l\right]})^2 \rangle_{\mathcal{C}_l}$ as a function of the fraction of entities belonging in the potential unconditional cooperator set. Since $\mu > \sigma^2$, $\lambda_{\mathcal{C}}$ is decisive in the determination for the level of cooperation, the networks with a larger eigenvalue will be better promoters of cooperation. In this aspect, we observe in Fig.~\ref{fig:complex-networks}b that the RR graph always has the largest eigenvalue, followed by the WS graph. As $\gamma$ increases, the differences in $\langle (x^{\left[\mathcal{C}_l\right]})^2 \rangle_{\mathcal{C}_l}$ become apparent, whereas $\lambda_{\mathcal{C}}$ slowly converges to the same value for each network since in each random graph case as it becomes a stochastic matrix. Because $\langle (x^{\left[\mathcal{C}_l\right]})^2 \rangle_{\mathcal{C}_l}$ is related to the amplitude of the noise, the networks with lower value will be better promoters of cooperation. As depicted in Fig.~\ref{fig:complex-networks}c, the RR graph again is the graph with the largest magnitude of the observable, but its eigenvalue is already converged to the maximum value. Thus it is less affected by the properties of  $\langle (x^{\left[\mathcal{C}_l\right]})^2 \rangle_{\mathcal{C}_l}$. We arrive at similar conclusions when inspecting the properties of the other graphs. We hereby note that in a structure where the noise amplitude has a larger magnitude,  $\langle (x^{\left[\mathcal{C}_l\right]})^2 \rangle_{\mathcal{C}_l}$ will play a bigger role in determining the steady state cooperative behavior. 

\begin{figure*}[ht!]
\includegraphics[width=\textwidth]{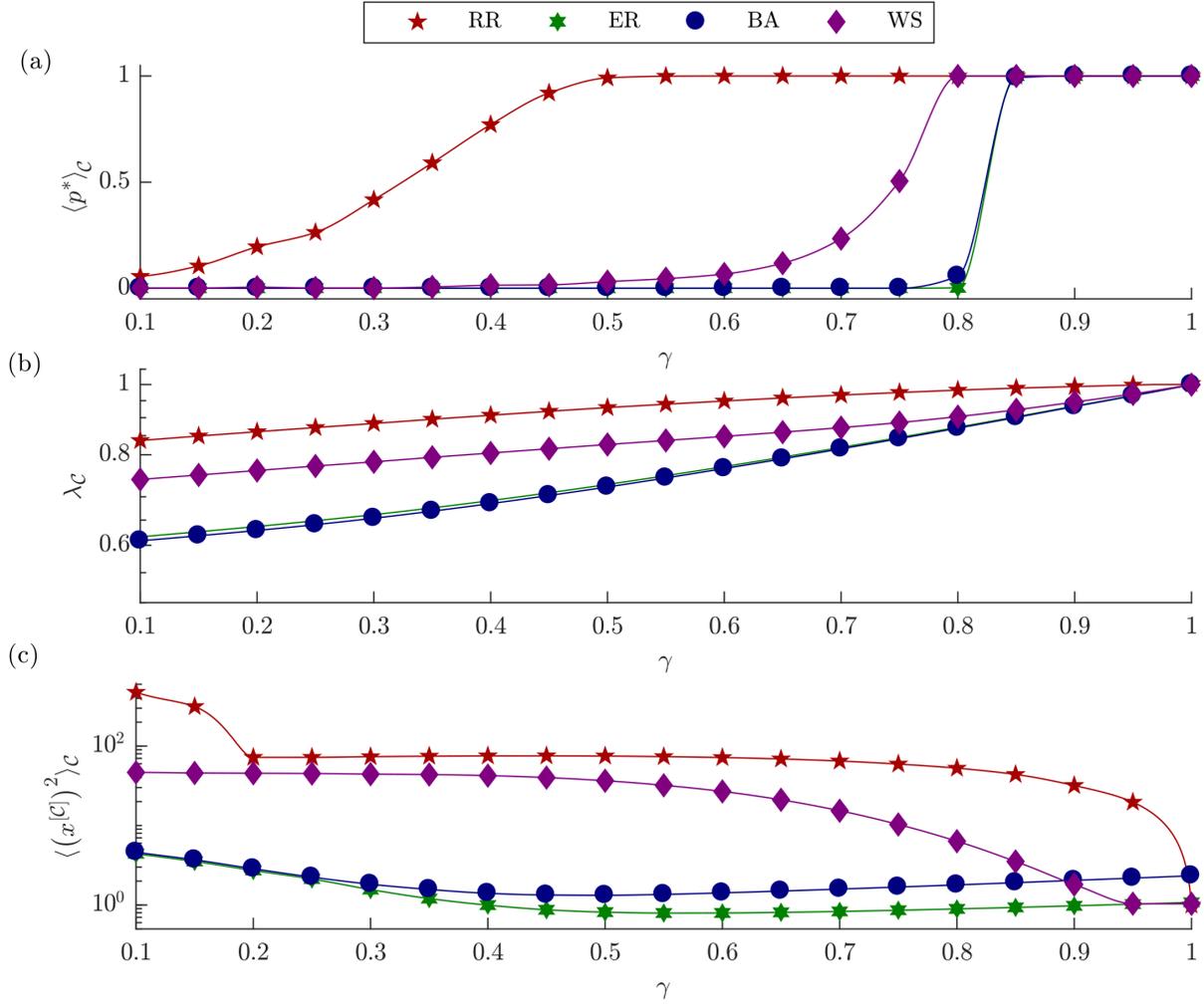}
\caption{\textbf{The role of complex networks in heterogeneous populations.} \textbf{(a)} Average steady state propensity for cooperation $\langle p^* \rangle_{\mathcal{C}}$ as a function of the fraction of potential unconditional cooperators $\gamma$ for RR, ER BA and WS graphs. \textbf{(b)} Average largest eigenvalue $\lambda_\mathcal{C}$ of the  matrix $A_\mathcal{C}$ as a function of $\gamma$ for the same networks. \textbf{(c)} Average $\langle (x^{\left[\mathcal{C}_l\right]})^2 \rangle_{\mathcal{C}_l}$ as a function of $\gamma$. In the simulation $\mu = 0.25$ and $\sigma^2 = 0.3$. The results are averaged across $100$ different unconditional defector choices and $100$ graph realizations with each graph having $2000$ nodes and an average degree of $4$. \label{fig:complex-networks}}
\end{figure*}

\textcolor{red}{To provide robustness of the network properties that we discussed in terms of the relationship between $\lambda_{\mathcal{C}}$ and $\langle (x^{\left[\mathcal{C}_l\right]})^2 \rangle_{\mathcal{C}_l}$ and $\gamma$, in Fig~\ref{fig:std-complex-networks} we provide the corresponding standard deviations, $S_{\lambda_{\mathcal{C}}}$ and $S_{\langle (x^{\left[\mathcal{C}_l\right]})^2 \rangle_{\mathcal{C}_l}}$ of the numerical results as a function of the fraction of entities belonging in the potential unconditional cooperator set. They have a significantly smaller magnitude than the averages shown in Fig.~\ref{fig:complex-networks}, thus indicating that the averaged results represent an adequate representation for the typical behavior of the graphs.}

\begin{figure*}[h]
\includegraphics[width=8.7cm]{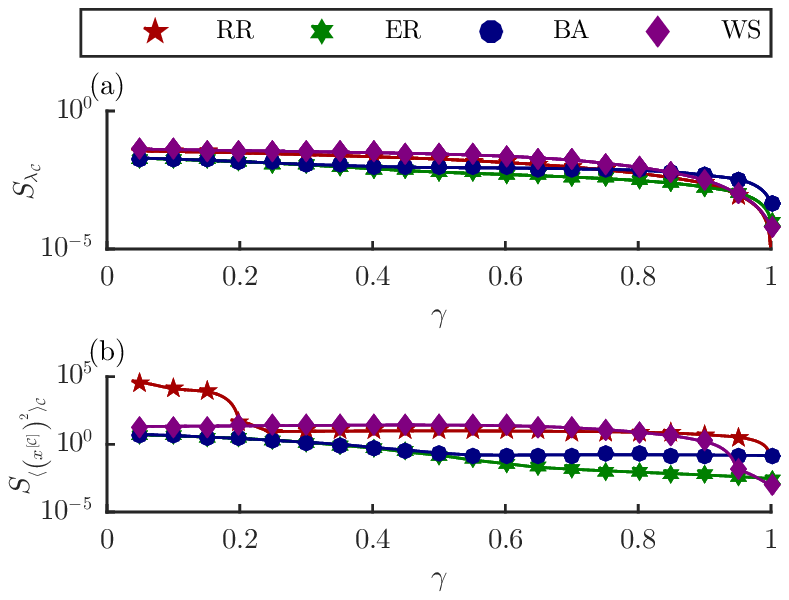}
\caption{\textbf{Standard deviation of the network properties.}  \textbf{(a)} Standard deviation $S_{\lambda_{\mathcal{C}}}$ of the largest eigenvalue $\lambda_\mathcal{C}$ of the matrix $A_\mathcal{C}$ as a function of $\gamma$ for the same networks. \textbf{(c)} Standard deviation $S_{\langle (x^{\left[\mathcal{C}_l\right]})^2 \rangle_{\mathcal{C}_l}}$ of $\langle (x^{\left[\mathcal{C}_l\right]})^2 \rangle_{\mathcal{C}_l}$ as a function of $\gamma$. In the simulation $\mu = 0.25$ and $\sigma^2 = 0.3$. The results are averaged across $100$ different unconditional defector choices and $100$ graph realizations with each graph having $2000$ nodes and an average degree of $4$. \label{fig:std-complex-networks}}
\end{figure*}

\section{Generalized reciprocity}
\label{sec:generalized-reciprocity}

So far, we addressed evolutionary behavior in which pooling and sharing is the sole cooperative mechanism. We showed that, as a consequence of the non-ergodicity, if certain conditions are satisfied, unconditional cooperation can be evolutionary stable even without the presence of any additional decision making mechanisms. In fact, the presence of additional auxiliary mechanisms should yield dynamics that at least complement the evolution of cooperation. To provide an initial insight on the role of these mechanisms, here we examine the cooperative behavior in the presence of a state-based generalized reciprocity update rule.

Generalized reciprocity suggests that cooperation can emerge as a consequence of previous positive experience with not necessarily the same group of opponents, i.e. it is based on the rule of ``help anyone if helped by someone''. This is significantly different from the two main forms of reciprocity, direct and indirect, which explain the emergence of cooperation either as a result of repeated encounters between the same group of entities or as an attempt to build positive reputation for future interactions~\cite{axelrod1981evolution,alexander2017biology}.

The main presumption which favors generalized reciprocity over other reciprocal mechanisms is that individual entities following this rule can be said to exhibit a simple state-based behavior. Due to this straightforward behavioral update rule, generalized reciprocity has been observed in a wide range of animal and human societies each manifesting different level of cognitive prowess and interacting in various environments~\cite{Rutte-2007,Leimgruber-2014,Gfrerer-2017,Bartlett-2006,Stanca-2009}.

While the extent to which generalized reciprocity is able to evolve as a sole cooperation mechanism has been a subject to an active debate, recent studies have shown that once this mechanism is present in a system, it induces dynamics which assist the stability of cooperation~\cite{pfeiffer2005evolution,Rankin-2009,vanDoorn-2012,nowak2006upstream,ito2019evolution}. Concretely, in~\cite{utkovski2017promoting,stojkoski2018role} it was argued that state-based generalized reciprocity effectively prevents the individual entities from being exploited by unconditional defectors, whereas in~\cite{stojkoski2018cooperation} it was suggested that this prevention is accompanied with maximization of the level of cooperation displayed by each entity. Nevertheless, the theoretical work done so far has assumed that the entities interact in an additive environment where ensemble averages are a good approximation for the stochastic behavior. In this section we extend the rule to account for possible resource dynamics driven by a multiplicative noise.

\subsection{Behavioral update}

To study the individual behavior under a simple generalized reciprocity rule we consider an update based on the entity's estimate for its growth rate
\begin{align}
    p_i\left(t+\Delta t\right) &= f_{i,t}\left[g_i(y_i(t),t)\right],
    \label{eq:update-rule}
\end{align}
where $f_{i,t}: \mathbb{R} \to \left[ 0,1 \right]$ is monotonic (nondecreasing). A plausible choice would be the sigmoid (logistic) function
\begin{align}
    f_{i,t}(\omega) &= \left[ 1 + \exp{(-\kappa_i(t))(\omega - \omega_i)} \right]^{-1},
    \label{eq:logistic-function}
\end{align}
where the midpoint $\omega_i$ is given by the steady state of $g_i(y_i(t),t)$ without pooling, i.e., $g^{\mathcal{D}}_i$. 

We remark that we purposefully allow for the steepness of the sigmoid function to be an unbounded monotonically increasing function of $t$, so as to account for the time-dependence in the variance of $g_i(y_i(t),t)$. For simplicity, we focus on the special case when
$\kappa(t) = t^{\alpha} $ where $\alpha$ is a positive parameter that captures the learning rate of the entities. In particular, $\alpha < 1$ corresponds to convergence towards equilibrium propensities to cooperate at a rate lower than the elapsed time $t$, whereas $\alpha > 1$ provides the opposite dynamics. The variable $\kappa(t)$ in our model is directly related to the intensity of choice parameter used in standard Choice theory (it is the analogue of the inverse temperature in statistical mechanics)~\cite{moran2020force}.

We point out that the introduced rule provides a simple description for the cooperative behavior in a wide range of interaction structures. The advantage of the behavioral update lies in its simplicity since~(\ref{eq:update-rule}) implies that an entity only has to know its current amount of resources in order to determine the next action. This is significantly different from other behavioral update rules where entities are required to optimize over domains depending on both opponent behavior and possible future interactions~\cite{stojkoski2018role}.

\subsection{Model properties}
\label{model-properties}

In the analysis performed in Section~\ref{sec:evolutionary-stability} we always ended up with dynamics that have a steady state. However, the introduction of the state-based generalized reciprocity rule may in fact disturb this property of the model and we may end up with complex dynamics whose study is out of the scope of the paper.

To greatly ease the analysis of generalized reciprocity here we consider two situations. In the first situation, we examine the individual cooperative behavior under the assumption that there is a steady state growth rate $g^*_i$ for every entity $i$. In the second situation, we study the properties of the model numerically and derive an analytical solution for the model in the circumstance when $\alpha <1$ since, as it will be shown, there is always a steady state.

\paragraph{Cooperative behavior:}
Let us assume that there is a steady state cooperative behavior $p^*_i$ and growth rate $g^*_i$ for each entity $i$ conforming to the rule~(\ref{eq:update-rule}). In addition, we assume that the growth rate of each entity is set such that the system is non-degenerate. By non-degenerate we mean that $g^\mathcal{D}_i \neq g^{\mathcal{C}}$ for all $i$. Then, we can derive the following properties of the model

\begin{enumerate}
\item[\textit{i.}] \textit{Prevention of exploitation:} -- It can be easily shown that $g^*_i < g^{\mathcal{D}}_i$ is an impossible situation for any entity following the behavioral update rule. To see this, assume that entity $i$ follows~(\ref{eq:update-rule}) and has $g^*_i < g^{\mathcal{D}}_i$. Then, the behavioral update rule indicates that the steady state propensity for cooperation of this entity is $p^*_i = 0$. Subsequently, this implies that 
\begin{align}
    y_i(t+\Delta t) \geq  y_i(t) \left[ 1 + \mu_i \Delta t + \sigma_i \varepsilon_i(t) \sqrt{\Delta t} \right].
    \label{eq:defector-inequality}
\end{align}
Taking the limit as $t \to \infty$ we get that $g^*_i \geq g^{\mathcal{D}}_i$, thus contradicting our initial assumption. This is a favorable property of state-based generalized reciprocity which has been also observed in additive dynamics~\cite{utkovski2017promoting}. It significantly differs from other forms of generalized reciprocity since it has been shown that they are prone to exploitation~\cite{nowak2006upstream}.
\item[\textit{ii.}]  \textit{Sufficient condition for existence of unconditional cooperators:} -- The behavioral update rule coupled with the monotonicity and unboundedness of $\kappa_i(t)$ imply that if $g^*_i > g^{\mathcal{D}}_i$ then $p^*_i = 1$. Moreover, it follows that a necessary condition for $p^*_i < 1$ is $g^*_i = g^{\mathcal{D}}_i$.
    
\end{enumerate}

Altogether, in steady state the entities may thus be attributed to two (disjoint) sets, $\mathcal{D}_{gen} = \left\{ d \in \mathcal{N} : p^*_d<1  \right\}$ and $\mathcal{C}_{gen} = \left\{ c \in \mathcal{N} : p^*_c=1  \right\}$, depending on the steady state propensity for cooperation $p^*_i$. The entities in $\mathcal{D}_{gen}$ are further characterized by $g^*_d=g^{\mathcal{D}}_d$, while the entities $y_c$ in $\mathcal{C}_{gen}$ have an unknown growth rate $g_c \geq g^{\mathcal{D}}_c$, which is dependent on the network parameters. We will refer to the entities in the sets $\mathcal{D}_{gen}$ and $\mathcal{C}_{gen}$ as ``generalized defectors'', respectively ``generalized cooperators'', with an intention to emphasize their role in the pooling and sharing mechanism.

\paragraph{Numerical experiment:}
The behavior of the individual growth rates ultimately depends on the magnitude of $\alpha$. 
This parameter indicates the speed at which every entity reacts to the environment: if an entity has a smaller-than-expected growth, it will converge towards unconditional defection faster (proportionally with $\alpha$), and otherwise towards unconditional cooperators. Even if the fraction of resources that is pooled and shared is small, the amount may be non-negligible for some of the entities which behave as unconditional defectors in evolutionary sense and, in fact, it may significantly increase their growth rate, thus making them ``generalized cooperators''. In other words, while the rule asserts that the entities which experience lower growth than their own eventually end up as generalized defectors, the number of generalized defectors may in general be smaller than the number of defectors inferred through evolutionary stability analysis.

To illustrate this effect we make use of the network given in Fig.~\ref{fig:network-cases}. Concretely, we initialize the drift and amplitudes of the entities so as under evolutionary analysis $p^*_3 = 0$ and cooperation is unstable for every other entity. Then, we
simulate the coupled dynamics of Eq.~(\ref{eq:network-discrete}) and Eq.~(\ref{eq:update-rule}) for $2 \times 10^4$ time steps and record the propensities for cooperation among the entities in the potential generalized cooperator set at the last point in time.

Fig.~\ref{fig:generalized-reciprocity}a provides a boxplot for the propensities for cooperation for each entity. In general, we observe two different regimes depending on $\alpha$. The first appears when $\alpha < 1$. In this case, the propensity to cooperate $p_i$ for each entity $i$ appears to be at a value less than $1$ but larger than $0$. As $\alpha$ increases, the propensities to cooperate for each $i \neq 3$ also increase, eventually converging to $p_i = 1$. In contrast, the propensity of entity $3$ decreases and converges to $p_3 = 0$. We assert that for the lower values of $\alpha$ the system is not in steady state and is eventually converging to the steady state where $p^*_i = 1$ for all $i \neq 3$, $p_3^* = 0$ and $g^*_i = g_3^{\mathcal{D}}$. This is evidenced in Fig.~\ref{fig:generalized-reciprocity}b where we plot the numerically observed growth rate for each entity as a function of $\alpha$. Concretely, we notice that the observed growth rate of each $i \neq 3$ is larger than their defector growth rate, thus implying that $p^*_i= 1$. In contrast, the growth rate of $3$ is smaller than its defector growth rate. As a consequence, the propensity to cooperate of this entity must be converging to $0$. This slow convergence is illustrated in Fig.~\ref{fig:generalized-reciprocity}c where we visualize the dynamics over time of each $p_i(t)$ for various $\alpha$.

The second regime occurs when $\alpha > 1$. In this situation we observe that the propensities for cooperation for each entity $i$ vary from one simulation to another. One explanation for this phenomena is the fact that when $\alpha \geq 1$ the steepness of the logistic curve diverges with a rate that is at least as fast as the divergence of the resources of each entity. This reduces the domain of all $p_i(t)$ to $\left\{ 0, 1\right\}$ faster. Since growth rate of each entity is highly stochastic we may observe non-equilibrium dynamics even though we numerically find that on average the growth rate of each potential generalized cooperator is larger than expected (Fig.~\ref{fig:generalized-reciprocity}b).

\begin{figure*}[t!]
\includegraphics[width=\textwidth]{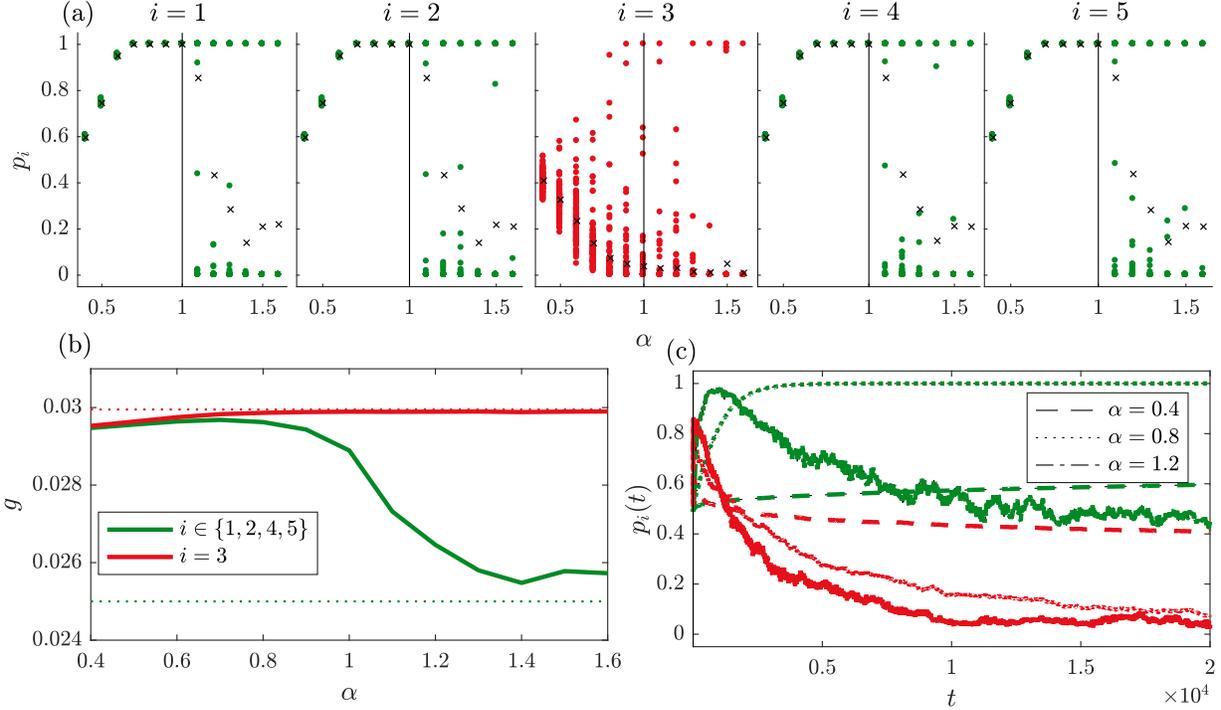}
\caption{\textbf{State-based generalized reciprocity in fluctuating environments.} \textbf{(a)} Observed propensities for cooperation $p_i(t)$ for each entity at $t = 2 \times 10^4$ over $100$ realizations. The circles are the results from each realization, whereas the crosses are the estimated averages across realizations for a given $\alpha$. \textbf{(b)} Numerically observed growth rate averaged across realizations. \textbf{(c)} Dynamics over time for the average $p_i(t)$ across realizations. The green lines are the averages for entities $i \in \{1,2,4,5 \}$ whereas the red lines are the averages for $i = 3$. \textbf{(a-c)} The parameters are set to $\mu_i = 0.03$ for all $i$, $\sigma = \sqrt{0.01}$ for $i \neq 3$ and $\sigma = 0.01$ for $i = 3$.
\label{fig:generalized-reciprocity}}
\end{figure*}

\paragraph{Analytical solution:}\label{Analytical-solution}The general case of $\alpha < 1$ in any interaction structure can be easily analytically solved and thus the observations explained. To see this, let $i$ represent the strongest entity in the network in the sense that $g_i^\mathcal{D} > g^\mathcal{C}$ and $g_i^\mathcal{D} > g_j^\mathcal{ D} $ for all $j \neq i$. 
We claim that the limit of the growth of each entity is at least as large as $g_i^{\mathcal{D}}$. 
For every finite time $t$ we approximate the growth rate $g_j(t)$ of each entity $j \neq i$ and notice that this quantity is bounded from below by the the growth induced by the shared resources from the strongest entity. Since the defection is slow enough, these shared resources are enough to bring the growth of each $j \neq i$ to $g_i^{\mathcal{D}}$. The formal proof is given in the Appendix. 

The previous discussion leads to the situation where each entity exhibits a growth rate equal to the growth rate of the strongest entity $g_i^{\mathcal{D}}$. In turn, from the cooperative behavior properties \textit{i.} and \textit{ii.} we infer that $p^*_j = 1$ for all $i \neq j$ and $p^*_i < 1$.

\section{Conclusion}

We investigated the evolution of cooperation in networked heterogeneous fluctuating environments. We found out that the fluctuations induce evolutionary behavior which may lead to emergence of components within the population structure. The properties of each component are characterized solely by the interaction topology and the traits possessed by the entities belonging in it. Thus, we may observe great disparities in the owned resources between cooperating entities belonging to different components.

Moreover, by introducing a simple behavioral update rule we showed that state-based generalized reciprocity enhances the promotion of cooperation in fluctuating environments. The rule induces dynamics under which each entity is prevented from exploitation, i.e., each $i$ observes a growth rate that is at least the same size as its own growth when behaving as an unconditional defector. More importantly, by construction, the entities whose observed growth rates at time $t$ that are greater than in the unconditional defector situation, also display greater propensities to cooperate. For the regime in which the learning rate is slower than the temporal dynamics, we analytically derived the exact behavior of each entity. When the learning rate is faster than the temporal dynamics we numerically observed complex dynamics. We argue that the observed dynamics is a consequence of the noise generating process having a larger impact on the evolutionary result and may produce situations with no stable state. In these situations, there is a never ending cycle in which the behavior of each entity alternates between unconditional cooperation and unconditional defection. While a detailed study of the dynamics is required, we believe that its implications are non-trivial and may contribute in explaining events such as the co-evolution of life cycles and multicellularity~\cite{hammerschmidt2014life,rainey2010cheats,pichugin2015modes}.

The implementation of our results goes beyond explaining the evolution of cooperation. In particular, the introduced rule is directly related to the concept of novelty search where individual entities decide their next actions on the basis of previous experience~\cite{lehman2008exploiting}. Novelty search is omnipresent in reinforcement learning and has been utilized in developing machines that efficiently mimic human behavior. In this aspect, we argue that the results discovered here behave as a building block in constructing machines which interact in a fluctuating environment.

\section{Appendix}
\label{appendix}
Here we provide the mathematical details of our claim in Section~\ref{model-properties} c. Analytical solution.
\begin{theorem} Let $y_k(t)$ be a solution to Eq.~(\ref{eq:network-differential}) with $\alpha <1$ and assume that $g_k^*$ exists for all $k$. In addition, assume that $i$ represents the strongest entity in the sense that $g_i^\mathcal{D}> g^\mathcal{C}$ and $g_i^\mathcal{D} > g_j^\mathcal{D}$ for all $j \neq i$. Then $g_k^* = g_i^\mathcal{D}$ for all $k$. Moreover, $p_j^*=1$ for $j \neq i$ and $p_i^* <1$. 
\end{theorem}
\begin{proof}
We argue by contradiction. Namely, let $g_j^* < g_i^\mathcal{D}$, and let entities $i$ and $j$ take part in a common pool, i.e. $A_{ji}>0$ and $A_{ij}>0$.
For any $T>T_0>0$  the solution to~\eqref{eq:network-differential} can be written as
\begin{align*}
     y_j(T) &= \int_{T_0}^T A_{ji}p_i(t)y_i(t)\mathrm{d}t - p_j(t)y_j(t)\mathrm{d}t 
     + \left(\int_{T_0}^T G_{T,T_0}(t)\right),
\end{align*}
where $G_T(t)$ is given as
    \begin{align*}
     G_{T,T_0}(t) &= \frac{y_j(T_0)}{T-T_0} + \sum_{k \neq i} A_{jk}p_k(t)y_k(t)(\mathrm{d}t+\mu_k \mathrm{d}t + \sigma_k \mathrm{d}W_k )  + \\ 
     &+A_{ji}p_i(t)y_i(t)(\mathrm{d}t+\hat{\mu}_i \mathrm{d}t + \sigma_i \mathrm{d}W_i ) + (1-p_j(t))y_j(t)(\mu_j\mathrm{d}t+ \mathrm{d}W_j).
    \end{align*}
The integral of this expression gives the total amount of resources of entity $j$ gathered from sources other than $i$, as well as resources from $i$ with a changed drift, $\hat{\mu}_i:= \mu_i -1$ between times $T_0$ and $T$. Therefore it is non-negative. We can use this fact to arrive at a lower bound of $g_j(T)= \frac{1}{T} \log(y_j(T))$ for $T-1>T_0\gg0$,
\begin{align*}
    g_j(T) &\ge \frac{1}{T} \log \left(   \int_{T_0}^T \left(  A_{ji} p_i(t)y_i(t) - p_jy_j(t) \right) \mathrm{d}t\right)\\
    %\big(-p_j \ge -1 \big) \text{  }
    &\ge \frac{1}{T} \log \left(   \int_{T_0}^T \left( A_{ji} p_i(t)y_i(t) - y_j(t) \right) \mathrm{d}t\right ),  
\end{align*}
with the second inequality following from the fact that $p_j \le 1$. The fact that the logarithms of the two integrals are well defined stems from the fact that because $g_j<g_i^\mathcal{D}$ we have $A_{ji} p_i(t)y_i(t) - y_j(t) >0$ for all $t>T_0$ with $T_0$ set to be large enough. The details are analogous to the ones in the justification of the passage to the limit that can be found below in inequality (\ref{eq:passage-to-limit}). \newline
\indent We can further relax the inequality by applying the Gronwall-Bellman lemma~\cite{Khasminskii, Touzi} (or, alternatively, the assumption of the existence of a finite $g_i^*$) to the properties of the update rule~(\ref{eq:update-rule}) of $p_i(t)$. Formally, the lemma implies that the norm $||z(t)||$ of a solution to a differential equation $dz = A(z(t),t)zdt+B(z(t),t)zdW$ with continuous bounded coefficients $|A|+|B|<M $ is bounded from above by a function of the form $t \mapsto \theta_1  \exp( \theta_2 t)$. In other words, $||z(t)|| \le \theta_1  \exp( \theta_2 t)$ for some coefficients $\theta_1, \theta_2 \ge 0$ and all $t>0$. Since the differential equation defining $ y(t)$ is with bounded coefficients, we can apply the lemma to $y_i(t)$ to deduce that $g_i(t)$ is also bounded. In our case we have 
\begin{align*}
    p_i(t) &= \frac{1}{1+\exp[-t^{\alpha}(g_i(t) - g_i^\mathcal{D})]} \ge \frac{1}{\xi+ \exp(\eta t^\alpha)},
\end{align*}
with constants $\xi,\eta \ge 0$ dependent only on the realization. This reduces the lower bound to
\begin{align*}
    g_j(T) &\ge \frac{1}{T}\log \left(  \int_{T_0}^T \left( \frac{ A_{ji}}{\xi+ \exp(\eta t^\alpha)}y_i(t) -y_j(t) \right) \mathrm{d}t\right )\\
  &\ge \frac{1}{T}\log \left(   \int_{T-1}^T \left(\frac{A_{ji}}{\xi+ \exp(\eta t^\alpha)}y_i(t) - y_j(t) \right) \mathrm{d}t\right ).
\end{align*}
Next, by applying Jensen's inequality to $x \mapsto \log(x)$ we further get that
\begin{align}
 g_j(T) &\ge \frac{1}{T}  \int_{T-1}^T\log \left( \frac{A_{ji}}{\xi+ \exp(\eta t^\alpha)}y_i(t) -y_j(t) \right) \mathrm{d}t.
\label{eq:gen-rep-inequality}
\end{align}
By taking the limit of this expression we get that
\begin{align*}
   g_j^* &\geq g_i^*, \\
       &\geq g_i^{\mathcal{D}},
\end{align*}
with the last inequality following from the prevention of exploitation property described previously. Altogether, this is the desired contradiction. Since the graph is connected, the property $g_j^* \ge g_i^\mathcal{D}$ is propagated to each entity. Now if $g_i^* > g_i^{\mathcal{D}}$, we would have $p_i^* =1$, and since all $p_j^*=1$ by $g_j^*>g_j^\mathcal{D}$, the growth $g_i^*$ would be equal to $g^\mathcal{C}$, which is impossible since $g_i^* \ge g_i^\mathcal{D}>g^\mathcal{C}$ where the first inequality is the prevention of exploitation property.
\end{proof}
Let us finish by justifying the passage to the limit of the right hand side in the inequality (\ref{eq:gen-rep-inequality}). First notice that for $k \in \{i,j\}$ we have $y_k(t) = \exp(t g_k^* + \varphi_k(t))$ for some functions $ \varphi_k(t)$ such that $\lim_{t \to \infty} \frac{\varphi_k(t)}{t} =0$. Since $g_i^* \ge g_i^\mathcal{D}$, we have, 
\begin{equation}
\begin{split}
     y_i(t) &= \exp(t g_i^* + \varphi_i(t)) \\
     &\ge  \exp(t g_i^\mathcal{D} + \varphi_i(t)), \\
     &\ge \exp(tg_j^* + t \varepsilon + \varphi_i(t)), \\
     &\ge_{+ \infty} \exp(tg_j^* + t\frac{\varepsilon}{2}+ \varphi_j(t)),\\
     &= \exp(t\frac{ \varepsilon}{2}) y_j(t),
\end{split}
\label{eq:passage-to-limit}
\end{equation} 
for all large $t$, and for any small $\varepsilon>0$ so that $g_i^\mathcal{D} > g_j^*+ \varepsilon$. Indeed, there exists $t_{\varepsilon}>0$ such that $t\frac{\varepsilon}{2} \ge |\varphi_j(t)| + | \varphi_i(t)|$ for all $t> t_{\varepsilon}$. In turn, for large $T$, $$\frac{y_j(T)(\xi+ \exp(\eta T^{\alpha}))}{y_i(T)A_{ji}} \le \frac{\xi +\exp(\eta T^{\alpha})}{A_{ji} \exp(T\frac{ \varepsilon}{2})} \xrightarrow[T \to \infty]{} 0,$$ implying that  
    \begin{align*}\frac{1}{T}\log \left( \frac{A_{ji}}{\xi+ \exp(\eta T^\alpha)}y_i(T) -y_j(T) \right)&= \\
    =& \overbrace{\frac{1}{T} \log \left( \frac{A_{ji}}{\xi+ \exp(\eta T^\alpha)}y_i(T)\right)}^{\xrightarrow[T \to \infty]{} g_i^*} +
    \overbrace{\frac{1}{T}\log \left( 1 -\frac{y_j(T)(\xi+ \exp(\eta T^{\alpha})}{y_i(T)A_{ji}} \right)}^{\xrightarrow[T \to \infty]{} 0}. 
    \end{align*} 
Finally, applying the change of variables $t' := t-T+1$ we obtain the expression
    \begin{align*}
        &\frac{1}{T}  \int_{T-1}^T\log \left( \frac{A_{ji}}{\xi+ \exp(\eta t^\alpha)}y_i(t) -y_j(t) \right)\mathrm{d}t  \\
        &=\int_{0}^1  \frac{1}{T}\log \left( \frac{A_{ji}}{\xi+ \exp(\eta (t' +T-1)^\alpha)}y_i(t' +T-1) -y_j(t' +T-1) \right)\mathrm{d}t',
    \end{align*}
    to which we can apply the dominated convergence theorem  ~\cite{Durett} as $T \to \infty$ thereby completing the justification of the passage to the limit.
    
%\bibliography{coop-bib2}  

\end{document}